\newif\ifmp
\mptrue %last one wins, uncomment or comment second one
\mpfalse

\ifmp
	\RequirePackage{fix-cm}
	\documentclass[smallextended]{svjour3}       % onecolumn (second format)
	\smartqed  % flush right qed marks, e.g. at end of proof
	\usepackage{graphicx}
\else
	\documentclass[11pt, oneside, A4paper]{article}
\fi

\usepackage{amsmath}
\usepackage{amsfonts}
\usepackage{amssymb}
\usepackage{graphicx}
\usepackage{color}
\usepackage{rotating}
\usepackage{url}
\usepackage{xspace}

\ifmp
\else
	\usepackage{amsthm}
\fi

\ifmp
\else
\fi

\newcommand{\R}{\mathbb R}
\newcommand{\Z}{\mathbb Z}

\newcommand{\F}{\mathcal{F}}

\newcommand{\K}{\mathcal{K}}
\newcommand{\qu}{\mathcal{Q}}
\newcommand{\smallc}{\textcolor{black}{$\F$-small complexity}\xspace}
\newcommand{\smallcc}[1]{\textcolor{black}{{#1}-small complexity}\xspace}

\newcommand{\conv}{\mathop{\mathrm{conv}}}
\newcommand{\cone}{\mathop{\mathrm{cone}}}
\DeclareMathOperator{\intcone}{int.cone}
\DeclareMathOperator{\rec}{rec}

\DeclareMathOperator{\C}{\mathcal{C}}
\DeclareMathOperator{\h}{\mathcal{H}}
\DeclareMathOperator{\argmin}{argmin}

\ifmp
\else
	\newtheorem{claim}{Claim}
	\newtheorem{proposition}{Proposition}

	\newtheorem{lemma}[proposition]{Lemma}
	\newtheorem{theorem}[proposition]{Theorem}
	\newtheorem{corollary}[proposition]{Corollary}
	
\fi

\newcounter{mynotes}
\newcommand{\remove}[1]{}

\newenvironment{cpf}{\begin{trivlist} \item[] {\em Proof of claim.}}{
$\diamond$
                       \end{trivlist}}

\begin{document}

\ifmp
	\title{Mixed-integer Quadratic Programming is in NP}
	%\titlerunning{Short form of title} % if too long for running head
	
	\author{Alberto Del Pia \and Santanu S. Dey \and Marco Molinaro}
	%\authorrunning{Short form of author list} 

	\institute{Alberto Del Pia \at IBM T. J. Watson Research Center\\
	           \email{alberto.delpia@gmail.com}
						 \and
						 Santanu S. Dey \at School of Industrial and Systems Engineering, Georgia Institute of Technology\\
             \email{santanu.dey@isye.gatech.edu}           %  \\
             \and
             Marco Molinaro \at School of Industrial and Systems Engineering, Georgia Institute of Technology\\
             \email{molinaro@isye.gatech.edu}
}

	\date{Received: date / Accepted: date}

\else
	\title{Mixed-integer Quadratic Programming is in NP}
	\author{Alberto Del Pia, Santanu S. Dey, Marco Molinaro}
\fi

\maketitle

\begin{abstract}
\emph{Mixed-integer quadratic programming (MIQP)} is the problem of optimizing a quadratic function over points in a polyhedral set where some of the components are restricted to be integral. In this paper, we prove that the decision version of mixed-integer quadratic programming is in NP, thereby showing that it is NP-complete. This is established by showing that if the decision version of mixed-integer quadratic programming is feasible, then there exists a solution of polynomial size. This result generalizes and unifies classical results that quadratic programming is in NP~\cite{Vavasis90} and integer linear programming is in NP~\cite{BoroshTreybig1976,GathenSieveking1978,KannanMonma1978,papadimitriou1981}.

\ifmp
\fi

\end{abstract}

\section{Introduction}
\label{sec: intro}
\emph{Mixed-integer quadratic programming (MIQP)} is the problem of optimizing a quadratic function over points in a polyhedral set that have some components integer, and others continuous. More formally, a MIQP problem is an optimization problem of the form:
\begin{eqnarray}\label{MIQP}
\begin{array}{llr}
&\textup{min} &x^\top Hx + c^\top x \\
&\textup{s.t.}& Ax \leq b \\
&&x \in \Z^p \times \R^{n-p},
\end{array}
\end{eqnarray}
where $H \in \mathbb{Q}^{n\times n}$ and is symmetric, $c \in \mathbb{Q}^n$, $A \in \mathbb{Q}^{m \times n}$ and $b \in \mathbb{Q}^m$.
The decision version of this problem is:
Does there exist a feasible solution to $\F(H,c,d,A,b)$ where $\F(H,c,d,A,b)$ is the set of $x$ satisfying
\begin{eqnarray}\label{MIQPd}
\begin{array}{llr}
&&x^\top Hx + c^\top x + d \leq  0\\
&& x \in \C:= \{ x : Ax  \leq  b\} \\
&& x \in \Z^p \times \R^{n-p}.
\end{array}
\end{eqnarray}
The special case of MIQP when all variables are required to be integer ($p=n$) is called \emph{integer quadratic programming (IQP)}. It is well known that IQP is NP-hard. 
%It has apparently never been proved that the decision versions of IQP and MIQP lie in NP.
%The purpose of this paper is to provide that proof.
We show that the decision versions of IQP and MIQP lie in NP. Therefore, decision version  of IQP and MIQP are NP-complete. This result generalizes and unifies classical results that \emph{quadratic programming} is in NP~\cite{Vavasis90} and \emph{integer linear programming} is in NP~\cite{BoroshTreybig1976,GathenSieveking1978,KannanMonma1978,papadimitriou1981}.

Recently, Del~Pia and Weismantel~\cite{DelPiaW14} showed that IQP can be solved in polynomial time when $n = 2$. It is a major open question whether IQP can be solved in polynomial-time for fixed dimension.

%\bigskip
\subsection{Statement of result and discussion}
Given a rational vector/matrix, its \emph{complexity} is the bit-size of its smallest binary encoding.
The complexity of a rational polyhedron $P \subseteq \R^n$ is the smallest complexity of a matrix $[A~b]$ such that $P = \{x : Ax \le b\}$.
%The complexity of $\F(H,c,d,A,b)$ is the complexity of the set of vectors/matrices $\{H, c, d, A, b\}$. 
The complexity of a finite set of objects is the sum over the complexity of the constituents objects.
We will prove the following result.
\begin{theorem}\label{thm:npresult}
Let $n,p \in \mathbb{Z}_{++}$. Let $H \in \mathbb{Q}^{n \times n}$, $c \in \mathbb{Q}^n$, $d \in \mathbb{Q}$, $A \in \mathbb{Q}^{m \times n}$, $b \in \mathbb{Q}^m$, and let $\phi$ be the complexity of $\{H, c, d, A, b\}$. If $\F(H,c,d,A,b)$ is non-empty, then there exists $x^{0}\in \F(H,c,d,A,b)$ such that the complexity of $x^0$ is bounded from above by $f(\phi)$ where $f$ is a polynomial function.
\end{theorem}

Theorem~\ref{thm:npresult} directly implies the following.
\begin{corollary}\label{thm:nph}
The decision versions of IQP and MIQP are NP-complete.
\end{corollary}

\begin{proof}
Given a graph $G=(V,E)$ and an integer $k$, determining whether there is a cut of cardinality at least $k$ in $G$ is NP-complete~\cite{GarJohSto76,Kar72}.
It is well known that such problem can be written as the decision IQP problem
\begin{eqnarray*}\label{maxcut}
\begin{array}{llr}
&& \sum_{v_iv_j \in E} (x_i+x_j-2x_ix_j) \ge k \\
&& x_i \in \{0,1\}^n \qquad \forall v_i \in V.
\end{array}
\end{eqnarray*}
Hence any problem in NP can be polynomially transformed to a decision IQP.
Theorem~\ref{thm:npresult} proves that there is a polynomial-length certificate for yes-instances of decision MIQP, showing that decision IQP and MIQP are in NP. \hfill $\square$
\end{proof}

%We end this section with a discussion on the complexity status of a general version of decision IQP when one varies not only the number of variables, but also the number of quadratic inequalities. The result of Theorem \ref{thm:npresult}  is in contrast with several well-known negative results:
We end this section by contrasting the result of Theorem \ref{thm:npresult} with several well-known negative results when one considers a more general version of decision IQP by varing the number of quadratic inequalities.
\begin{enumerate}
\item `Many' general quadratic inequalities: 
%Theorem \ref{thm:npresult} implies that decision IQP is decidable. On the other hand, 
By using a simple reduction from the problem of determining the feasibility of a quartic equation in 58 non-negative integer variables, we obtain that determing the feasibility of a system with $2\left({58 \choose 2} + 58 + 1 \right) $ quadratic inequalities and $58$ linear inequalities in $\left({58 \choose 2} + 58 \right) $ continuous variables and 58 integer variables is undecidable (see Theorem 3.2 and Theorem 3.3(i) in~\cite{Koeppe2012}). Therefore already with 3424 quadratic inequalities, 58 linear inequalities, 58 integer variables, 1711 continuous variables, it is not possible to bound the size of smallest feasible solution.

\item Two general quadratic inequalities: In the presence of two quadratic inequalities (in fact one quadratic equation) in two variables, there exist examples (the so called \emph{Pellian} and \emph{anti-Pellian} equations) where the minimal binary encoding length of any feasible integral solution is exponential in the complexity of the instance~\cite{Lagarias80}. 

\item `Many' convex quadratic inequalities: Consider the following system of inequalities~\cite{Khachiyan1983}:
\begin{eqnarray}
x_1 & \geq &2 \\
x_{i} & \geq & x_{i-1}^2 \ \forall i\in\{2, \dots, n\} \\
x &\in& \mathbb{Z}^n.
\end{eqnarray}
It is clear that for this system, the minimal binary encoding length of any feasible integral solution is exponential in the complexity of the instance.
\end{enumerate}
%Is it an interesting open problem if one can prove a more general version of (\ref{MIQPd}) is in NP in the presence of more quadratic inequalities (of course, not by bounding the minimal binary encoding length feasible solutions and with lesser than $2\left({58 \choose 4} + {58 \choose 3} + {58 \choose 2} \right) $ quadratic inequalities).

The above examples serve to highlight the fact that the result of Theorem \ref{thm:npresult} (a feasible system of inequalities with exactly one quadratic inequality  always has at least one integer feasible solution of small size) is tight with respect to the number of quadratic inequalities.

The rest of the paper is organized as follows. Section 2 collects all notation and results that are needed to prove Theorem \ref{thm:npresult}. Section 3 presents a proof of Theorem \ref{thm:npresult}.

%Notation:
%Throughout this paper, we use the word ``size'' to denote the length of the usual binary encoding (see \cite{Schrijverbook} for more details), 
%##############################################################
%##############################################################
%##############################################################

\section{Preliminaries}

%\subsection{Complexity of objects}
\subsection{Notation}

Throughout this paper, we use 
$e^i$ to represent the $i$-th vector of the standard basis of $\R^n$,
$\textup{sign}(u)$ to represent the sign of a real number $u$, 
$\textup{dim}(S)$ to represent the affine dimension of $S$,
$\conv(S)$ to represent the convex hull of a set $S$,
$\textup{cone}(S)$ to represent the conic hull of a set $S$,
and $\textup{int.cone}(S)$ to represent the set $\sum_{r^j \in S}\lambda_j r^j, \lambda_j \in \mathbb{Z}_{+} \ \forall j$.

 Given an object $\mathcal{O}$ and another object $f(\mathcal{O})$ that is a function of it, we say that $f(\mathcal{O})$ has \emph{\smallcc{$\mathcal{O}$}} if the complexity of $f(\mathcal{O})$ is at most a polynomial function of the complexity of $\mathcal{O}$ (or more precisely, there is a polynomial $p$ such that for every input object $\mathcal{O}$, the complexity of $f(\mathcal{O})$ is at most $p(complexity(\mathcal{O}))$).

\subsection{Quadratic programming is in NP}

%A \emph{quadratic programming} (QP) problem is an optimization problem of the form
%$$\textup{min} \{ x^\top Hx + c^\top x : Ax \leq b, \, x \in \mathbb{R}^n\},$$
%where $H,c,A,b$ are defined like for MIQP.
%Therefore QP, and its decision version, are the special cases of the corresponding MIQP when $p=0$.
\emph{Quadratic programming (QP)} is the special case of MIQP when all variables are continuous ($p=0$).
%Consider the continuous relaxation of the set (\ref{MIQPd}). 
Vavasis~\cite{Vavasis90} proved that the decision version of QP is in NP.

\begin{theorem}\label{QPisNP}
The feasibility problem over the continuous relaxation of (\ref{MIQPd})
% \mmnote{Should we change \eqref{MIQPd} for \eqref{MIQP}?}
is in NP. Moreover, suppose that the continuous relaxation of (\ref{MIQP}) has a global optimal solution. Then there exists a system of rational linear equations of \smallcc{$\{H,c, d, A,b\}$} whose solution is one of the global optimal solutions. 
\end{theorem}

\subsection{Mixed-integer linear programming is in NP} \label{sec:IPNP}

We will need the following generalization of a classical result that can be used to prove that the decision version of mixed integer linear programming (MIP) is in NP.
%; since this result is crucial to the rest of the paper, we provide its proof.
%\begin{definition}[Simple cone]
We say that a pointed polyhedral cone $C \subseteq \mathbb{R}^n$ is a \emph{simple cone} if the number of extreme rays is equal to the the dimension of the cone. 
%\end{definition}

\begin{proposition}\label{IPNP2}
Let $P = \{x : Ax \le b\} \subseteq \mathbb{R}^{p+q}$ be a rational pointed polyhedron. Then there is a finite family $\{P_i\}_i$ of polytopes, and a finite family $\{R_K\}_{K \in \K}$ of subsets of extreme rays of $P$ with the following properties:

	\begin{enumerate}
		\item $P \cap (\Z^p \times \R^q) = \bigcup_{i, K \in \K} (P_i + \intcone(R_K))$;
		\item Each polytope $P_i$ and each vector in $R_K$ has \smallcc{$[A~b]$};
		\item For each $K \in \K$, all vectors in $R_K$ are integral;
		\item Each cone $\cone(R_K)$ is simple.
	\end{enumerate}
\end{proposition}

\begin{proof}
	Assume $P$ is non-empty, otherwise there is nothing to prove. By standard polyhedral theory, there is a set of vectors $\{v^1, \ldots, v^\ell\}$ (the vertices of $P$) and a set of \emph{integral} vectors $\{r^1, \ldots, r^m\}$ (a scaling of the the extreme rays of $P$) such that $P = \conv\{v^1, \ldots, v^\ell\} + \cone\{r^1, \ldots, r^m\}$ and the $v^i$'s and $r^j$'s have \smallcc{$[A~b]$} (see for example Chapters 7 and 10 in~\cite{Schrijverbook}).

For a subset $K \subseteq \{1, \ldots, m\}$, let $R_K = \{r^j : j \in K\}$ be the set of extreme rays indexed by $K$. Let $\K$ be the set of $K$'s such that the cone $\cone(R_K)$ is simple. Finally consider the set $B =\bigcup_{k \in \K} B^K$ where 
	\begin{align}
		 B^K = \bigg\{x \in \Z^p \times \R^q : x = v + \sum_{j \in K} \mu_j r^j, ~v \in \conv\{v^1, \ldots, v^\ell\}, ~\mu_j \in [0,1] ~\forall j \bigg\}.
	\end{align}
	Notice that $B$ is a union of the polytopes given by the fibers $\{\bar{y}\} \times \{z \in \R^q: (\bar{y}, z) \in B^K\}$ ranging over all $K \in \K$ and $\bar{y} \in B^K|_p$, where $B^K|_p$ is the projection of $B^K$ to the first $p$ coordinates. Using the fact that $v^i$'s and $r^j$'s have \smallcc{$[A~b]$} and $|K| \leq p + q$, we get that all points in $B^K|_p$ have \smallcc{$[A~b]$}. Hence each of these fibers also has \smallcc{$[A~b]$} since it is the intersection of two \smallcc{$[A~b]$} polyhedron $\big\{x \in \R^p \times \R^q : x = v + \sum_{j \in K} \mu_j r^j, ~v \in \conv\{v^1, \ldots, v^\ell\}, ~\mu_j \in [0,1] ~\forall j \big\}$ and $\{ x \in \R^p \times \R^q : x|_p = \bar{y}\}$. Let $\{P_i\}_i$ be this collection of fibers. 

%Also, for a subset $K \subseteq \{1, \ldots, m\}$, let $R_K = \{r^j : j \in K\}$ be the set of extreme rays indexed by $K$; finally let $\K$ be the set of $K$'s such that the cone $\cone(R_K)$ is simple.
	
	By construction, properties 2, 3 and 4 of the proposition are satisfied, so it suffices to show property 1. By using distributivity of union and Minkowski sums, notice $\bigcup_{i, K \in \K} (P_i + \intcone(R_K)) = B + \bigcup_{K \in \K} \intcone(R_K)$.
	
	To show $P \cap (\Z^p \times \R^q) \subseteq B + \bigcup_{K \in \K} \intcone(R_K)$, take a point $x \in P \cap (\Z^p \times \R^q)$. We can write it as $x = v + r$ for $v \in \conv(v^i)_i$ and $r \in \cone(r^j)_j$. Using Carath\'eodory's Theorem, we see that there exists $K \in \K$ such that the simple cone $\cone(R_K)$ contains $r$. Then consider multipliers $\mu_j \in \R_+$ for $j \in K$ such that $r = \sum_{j \in K} \mu_j r^j$. Breaking up the multipliers into their fractional and integer parts, we get that $$x = v + \sum_{j \in K} (\mu_j - \lfloor \mu_j \rfloor) r^j + \sum_{j \in K} \lfloor \mu_j \rfloor r^j.$$ Clearly the last term belongs to $\intcone(R_K)$. Moreover, this term is integer (since the $r^j$'s are integer) and $x \in \Z^p \times \R^q$, thus the remaining part $v + \sum_{j \in K} (\mu_j - \lfloor \mu_j \rfloor) r^j = x - \sum_{j \in K} \lfloor \mu_j \rfloor r^j$ belongs to $\Z^p \times \R^q$ and hence to $B$. Thus $x \in B + \intcone(R_K)$, concluding this part of the proof. 
	
	We now show the reverse direction $P \cap (\Z^p \times \R^q) \supseteq B + \bigcup_{K \in \K} \intcone(R_K)$. It is easy to see that $P \supseteq B + \bigcup_{K \in \K} \intcone(R_K)$, since $P = \conv(v^i)_i + \cone(r^j)_j$ and $B \subseteq \conv(v^i)_i + \cone(r^j)_j$ and $\intcone(R_K) \subseteq \cone(r^j)_j$. Also, $B \subseteq \Z^p \times \R^q$ and $\intcone(R_K) \subseteq \Z^{p + q}$ (again since the $r^j$'s are integral), and hence $B + \bigcup_{K \in \K} \intcone(R_K) \subseteq \Z^p \times \R^q$. This concludes the proof.  \hfill $\square$
\end{proof}

	One way of interpreting this decomposition is the following: Notice that each set $\intcone(R_K)$, for $K \in \K$, is linearly isomorphic to $\Z_+^{|K|}$; this proposition then asserts that we can decompose any mixed-integer linear set into (overlapping) sets that are affinely isomorphic to some $\Z_+^{n'}$. Resorting to the product structure of $\Z_+^{n'}$ will be instrumental in our main proof.

	Also, notice that this proposition proves that the decision version of MIP is in NP: 
\begin{proposition} \label{MIP in NP}
The decision version of MIP is in NP.
\end{proposition}
Indeed, if the decision version of a MIP is true, then one can present a vertex of one of the polytopes $P_i$ as a certificate.

\section{Proof of Theorem \ref{thm:npresult}}

We begin this section with some technical lemmas.

\begin{lemma}[Normalizing hyperplane]
\label{lem:hyperplane}
Let $C = \textup{cone}\{r^1, \dots, r^s\}\subseteq \mathbb{R}^n$ be a pointed cone.
Then there exists a hyperplane $\mathcal{H} = \{ x : f^\top x = 1\}$ such that:
\begin{enumerate}
\item The complexity of $f$ is polynomially bounded by the maximum complexity of $r^i$, for $i=1,\dots,s$;
\item If $x \in C$ and $|\!|x |\!| = 1$, then $f^\top x \geq \frac{1}{R}$ where $R = \textup{max}_{i \in \{1, \dots, s\}}\{|\!|r^i|\!|\}$.
\item $C \cap \mathcal{H}$ is bounded.
%\anote{We should probably add: 3. $C \cap \mathcal{H}$ is bounded.}
\end{enumerate}
\end{lemma}
\begin{proof}
Let $\{r^{s+1}, \dots, r^t\}$ be a minimal subset of $\{e^1,\dots,e^n\}$ with the property that the cone $C' = \cone\{r^1, \dots, r^s,r^{s+1}, \dots, r^t\}$ is full-dimensional. 
Clearly $|\{r^{s+1}, \dots, r^t\}|=n-d$, where $d$ is the dimension of $C$, and $C'$ is pointed.
Let ${f}$ be an extreme point of the following polyhedron: $\{w : w^\top r^i \geq 1 \ \forall i \in \{1, \dots, t\}\}$ (an extreme point exists since the rank of the matrix defining the polyhedron is $n$). Then $f$ satisfies the first and third criteria. Suppose $\hat{x} \in C$ and $\|\hat{x}\| = 1$. There exists $0 < \mu \leq 1$ such that $\mu\hat{x}$ belongs to the polytope defined as the convex hull of $\{\frac{r^i}{\| r^i\|} : i=1,\dots,s\}$ (since the maximum norm of any vector in this polytope is $1$). Thus there exist $\lambda_i, \ i=1,\dots,s$, with $\sum_{i=1}^s \lambda_i=1$ such that ${f}^\top \hat{x} \geq {f}^\top (\mu\hat{x}) = \sum_{i = 1}^s \lambda_i \frac{1}{|\!| r^i|\!|}{f}^\top r^i \geq \frac{1}{R}$.  \hfill $\square$
\end{proof}
Sometimes we will apply Lemma~\ref{lem:hyperplane} to a pointed cone $C$, without giving explicitly the set of rays $\{r^1, \dots, r^s\}$.
It is well known (see for example Theorem 10.2 in~\cite{Schrijverbook}) that facet and vertex complexity of a rational polyhedron are polynomially related.
Hence there exist vectors $r^1, \dots, r^s$, each of \smallcc{$C$}, such that $C = \textup{cone}\{r^1, \dots, r^s\}$.
Hence, in this case, Lemma~\ref{lem:hyperplane} implies that there exists a normalizing hyperplane $\mathcal{H} = \{ x : f^\top x = 1\}$ such that:
\begin{enumerate}
\item $f$ has \smallcc{$C$};
\item For every nonzero $x \in C$, there exists $\mu > 0$ such that $\mu x \in \mathcal{H}$.
\end{enumerate}

The following lemma outlines a crucial decomposition strategy for searching integer feasible points. 
\begin{lemma} \label{decomposition}
Let $C$ be a simple pointed cone such that $x^\top Hx \ge 0$ for every $x \in C$. Let $\mathcal{H}= \{x : f^{\top}x = 1\}$ be the normalizing hyperplane from Lemma \ref{lem:hyperplane}. Then there exists a finite family of simple cones $C^i$, $i\in I$ such that 
\begin{itemize}
\item[$(a)$] $\bigcup_{i\in I} C^i = C$,
\item[$(b)$] for every $i \in I$, if a face $F$ of $C^i$ satisfies $\textup{min} \{ x^\top Hx : x \in F \cap \mathcal{H} \} = 0$, then there exists an extreme ray $v$ of $F$ with $v^\top H v = 0$,
\item[$(c)$] for every $i \in I$, $C^i$ has \smallcc{\{H,C\}} and dimension of $C^i$ is equal to the dimension of $C$.
\end{itemize}
\end{lemma}

\begin{proof}
%Without loss of generality we may assume $C$ is full-dimensional.

The proof is by induction on the dimension $n$ of the cone. If the cone has dimension one, then the claim is trivially true.

Since $C \cap \mathcal{H}$ is a compact convex set, by Theorem \ref{QPisNP} there exists an optimal solution $\bar x$ of the problem $\textup{min} \{ x^\top Hx : x \in C \cap \mathcal{H}\}$ that has \smallcc{$\{H,C,\mathcal{H}\}$}.
As $\mathcal{H}$ has \smallcc{$C$}, $\bar x$ has \smallcc{$\{H,C\}$}.
If the minimum value is strictly positive, then the result is trivially true. So we now assume that the minimum value is zero. 

Let $F_i$, $i \in I$ be the facets of $C$ that do not contain $\bar x$.
By induction, for every $i \in I$, there exist finitely many simple cones (of dimension $n-1$, and with $n-1$ extreme rays) $C_i^j$, $j\in J(i)$ that satisfy (a) and (b) with respect to the $n-1$ dimensional cone $F_i$.
We show that the family of cones 
\begin{equation} \label{fam}
\{\cone\{C_i^j \cup \{\bar x\}\}: i \in I, \ j \in J(i)\}
\end{equation}
satisfies (a) and (b).
Since for every $i \in I$, the vector $\bar x$ is affinely independent from all the vectors in $F_i$, each element of \eqref{fam} is a simple cone. It is straightforward to verify that (a) holds. Condition (b) holds by induction for all the faces of $C_i^j, \ i \in I, \ j\in J(i)$, and it holds also for all the remaining faces of the cones in \eqref{fam} because they all contain $\bar x$ as an extreme ray.

The above proof of (a) and (b) can be seen as a constructive algorithm that recursively constructs the simple cones $C^i$, $i\in I$.
In order to show that (c) holds, we just need to prove that all the $n$ extreme rays of the cones constructed in such way have \smallcc{$\{H,C\}$}.
Note that such extreme rays are either extreme rays of $C$, in which case have \smallcc{$C$}, or optimal solutions of a problem $\textup{min} \{ x^\top Hx : x \in F \cap \mathcal{H} \}$, for a face $F$ of $C$, in which case have \smallcc{$\{H,C\}$}.  \hfill $\square$
\end{proof}

Now we are ready to present a proof of Theorem \ref{thm:npresult}.

\begin{proof}[Proof of Theorem~\ref{thm:npresult}]	
	
	Consider a feasible $\F = \F(H, c, d, A, b)$ and let $\qu(x) := x^T H x^T + c^T x + d$ denote the quadratic form. Without loss of generality, we assume that the polyhedron $\C := \{ x : Ax  \leq  b\}$ is pointed: If not, consider the partition of the feasible region problem into $2^n$ pieces as
\begin{eqnarray*}
x \in \F(H, c, d, A, b) \\
x_i \geq 0 \ i \in S \subseteq \{1, \dots, n\}\\
x_i \leq 0  \ i \in \{1, \dots, n\}\setminus S,
\end{eqnarray*}
for every $S \subseteq \{1, \dots, n\}$; note that the complexity of the additional constraints is $O(n)$ and therefore each part in this partition has \smallc, so we can restrict to a non-empty part. 
	
	Notice that an external description of $\rec(\C)$ can be obtained by an external description of $\C$ by replacing all the right-hand sides with a zero, hence $\rec(\C)$ has \smallc. Using our assumption that $\rec(\C)$ is pointed, let $\mathcal{H}: = \{x : f^{\top}x = 1\}$ be the normalizing hyperplane from Lemma \ref{lem:hyperplane} for $\rec(\C)$. Examine the optimization problem:
	\begin{align} \label{eq:pureQuad}
		\begin{split}
			\textup{min} ~~& r^\top H r\\
			\textup{s.t.} ~~& r \in \rec(\C) \cap \h
		\end{split}
	\end{align}
Since $\rec(\C) \cap \h$ is compact, there exists a global optimal value. We break up into two cases depending on the sign of the optimal value. 

\paragraph{Case 1: The optimum for \eqref{eq:pureQuad} is strictly negative.} We construct a feasible solution of $\F(H, c, d, A, b)$ as follows. Since $\rec(\C)$ and $H$ have \smallc, Theorem \ref{QPisNP} asserts that there is an optimal solution $r^*$ for \eqref{eq:pureQuad} which has \smallc. Then let $\tilde{r}$ be an \emph{integer} vector with \smallc obtained by scaling $r^*$. Also, by Proposition \ref{MIP in NP}, let $\tilde{x}$ be a point in the mixed-integer linear set $\C \cap (\Z^p \times \R^q)$ with \smallc.
%for example, take any $\bar{x} \in P \cap (\Z^p \times \R^q)$ and let $\tilde{x}$ be an optimal solution for $\max\{ 0 : x \in P, ~x_i = \bar{x}_i ~~\forall i \le p\}$ with \smallc, which exists from Theorem \ref{QPisNP}. 

	For every $\lambda \in \Z_+$, the point $\tilde{x} + \lambda \tilde{r}$ belongs to $\C \cap (\Z^p \times \R^q)$. Expanding the quadratic form (and giving names to the different terms):
	\begin{align}
		\qu(\tilde{x} + \lambda \tilde{r}) = \lambda^2 \tilde{r}^\top H \tilde{r} + \lambda \left(2 \tilde{x}^\top H \tilde{r} + c^\top \tilde{r} \right) + c^\top \tilde{x} + d := \lambda^2 v_1 + \lambda v_2 + v_3.
	\end{align}
	Since $v_1 < 0$, this is a strictly concave polynomial in $\lambda$, and so setting $\lambda$ larger than its larger root gives $\qu(\tilde{x} + \lambda \tilde{r}) < 0$. Explicitly, let $$\tilde{\lambda} = \max\left\{\bigg\lceil \frac{-v_2 - \sqrt{v_2^2 - 4 v_1 v_3}}{2 v_1} \bigg\rceil, 0 \right\}.$$ Then $\tilde{x} + \tilde{\lambda} \tilde{r}$ is feasible for \eqref{MIQPd}, and moreover it is easy to verify that it has \smallc. This concludes the proof of Case 1.

\paragraph{Case 2: The optimum for \eqref{eq:pureQuad} is non-negative.} Then let $\{P_i\}_i$ and $\{R_K\}_{K \in \K}$ be the decomposition of $\C \cap (\Z^p \times \R^q)$ from Proposition \ref{IPNP2}. By the guarantees of this decomposition, there is $\bar{i}$ and $\bar{K}$ such that $(P_{\bar{i}} + \intcone(R_{\bar{K}})) \cap \{x : \qu(x) \le 0\}$ is non-empty. Since $R_{\bar{K}}$ is simple and pointed (since we assume $\C$ pointed), we can use Lemma \ref{decomposition} to refine $R_{\bar{K}}$ into the family of cones $\{R_{\bar{K}, j}\}_j$; again due to its guarantees, let $\bar{j}$ be such that $(P_{\bar{i}} + \intcone(R_{\bar{K}, \bar{j}})) \cap \{x : \qu(x) \le 0\}$ is non-empty. We show that this set has a point of \smallc. To simplify the notation, let $P := P_{\bar{i}}$, and enumerate $R_{\bar{K}, \bar{j}} = \{r^j\}_j$.

	Now let $F := \cone(r^j)_j$. In addition, for an index $i$, we exclude ray $r^i$ to define the face $F_i := \cone(r^j)_{j \neq i}$, and similarly for a set of indices $J$, let $F_J := \cone(r^j)_{j \notin J}$. Finally, we define the $\intcone$ version of these cones, namely $F^I := \intcone(r^j)_j$, $F^I_i := \intcone(r^j)_{j \neq i}$ and $F^I_J := \intcone(r^j)_{j \notin J}$. So we are interested in the solutions to 
	\begin{align} \label{eq:lastPart}
		\begin{split}
		&\qu(x) \le 0\\
		& x \in P + F^I.
		\end{split}
	\end{align}

	We first analyze the behavior of $\qu$ over a single direction $r^j$. Any point in $P + F^I$ can be written as $x^i + \mu r^i$ for $x^i \in P + F^I_i$ and $\mu \in \Z_+$. Define $J := \{j \in \{1,\dots,n\} : (r^j)^\top H r^j = 0\}$. Then $\qu$ has linear behavior along the directions $r^i$ with $i \in J$: for all $x^i \in P + F^I_i$ and $\mu \in \Z_+$
	\begin{align} \label{eq:linearRj}
		\qu(x^i + \mu r^i) = \mu \cdot \left(2 (x^i)^\top H r^i + c^\top r^i\right) + (x^i)^\top H x^i + c^\top x^i + d ~~~~~~~~~\forall i \in J.
	\end{align}
	Hence, if there is $i \in J$ and a point $x^i \in P + F^I_i$ such that the first term $2 (x^i)^\top H r^i + c^\top r^i$ is negative, then we can find a large scaling $\mu$ such that the point $x^i + \mu r^i$ satisfies \eqref{eq:lastPart}; in fact, we can construct such a point in a way that it is has \smallc.
	
	\begin{claim} \emph{1} \label{claim one}
		Consider $i \in J$ and the linear optimization problem $\min\{2 (x^i)^\top H r^i + c^\top r^i : x^i \in P + F^I_i\}$. If the optimum of this problem is negative, then there is a point $\tilde{x}^i \in P + F^I_i$ which has \smallc and negative objective value.
	\end{claim}
	
	\begin{cpf}
		To simplify the notation, let $obj(x) = 2 x^\top H r^i + c^\top r^i$ denote the objective function. Let $\tilde{p} \in \argmin\{obj(p) : p \in P\}$. Since $P$ has \smallc, it follows that $\tilde p$ has \smallc. Clearly if $obj(\tilde{p}) < 0$, then set $\tilde{x}^i$ to $\tilde{p}$ as the desired point in $P + F^I_i$, concluding the proof. Otherwise, by linearlity of $obj$ and the definition of $F^I_i$ there exists some $j \neq i$ such that $2(r^j)^THr^i <0$. Then let $\tilde{\eta}_j$ be the smallest non-negative integer satisfying
		\begin{align*}
			&obj\bigg(\tilde{p} + r^j \tilde{\eta}_j\bigg) \le -1\\
			&\tilde{\eta}_j \in \Z_+ 
		\end{align*}
Clearly $\tilde{\eta}_j$ has \smallc and therefore $\tilde{x}^i = \tilde{p} + \tilde{\eta}_jr^j$ is the desired point in $P + F^I_i$, concluding the proof.
%
%
%		To simplify the notation, let $obj(x) = 2 x^\top H r^i + c^\top r^i$ denote the objective function. Suppose $\bar{x}^i \in P + F^I_i$ is such that $obj(\bar{x}^i) < 0$. We can write $\bar{x}^i = \bar{p} + \sum_{j \neq i} r^j \bar{\mu}_j$ for $\bar{p} \in P$ and $\bar{\mu}_j \in \Z_+$ for all $j \neq i$. By linearity of $obj$, letting $\tilde{p} \in \argmin\{obj(p) : p \in P\}$, we have that $v^* := obj(\tilde{p} + \sum_{j \neq i} r^j \mu_j)$ is at most $obj(\bar{x}^i)$ and hence negative. Moreover, since $P$ has \smallc, it follows that $\tilde p$ and $v^*$ also have \smallc.
%		
%		Then consider the set of solutions to the following integer linear programming (IP) problem over variables $\mu_j$'s:
%		%
%		\begin{align*}
%			&obj\bigg(\tilde{p} + \sum_{j \neq i} r^j \mu_j\bigg) \le v^*\\
%			&\mu_j \in \Z_+ ~~~~~~ \forall j \neq i.
%		\end{align*}
%		By construction this IP is feasible and has \smallc. Since linear IPs have solutions of small complexity (see Section \ref{sec:IPNP}), let $\{\tilde{\mu}_j\}_{j \neq i}$ be any such solution; then $\tilde{p} + \sum_{j \neq i} r^j \tilde{\mu}_j$ is the desired point in $P + F^I_i$, concluding the proof. 
	\end{cpf}

	Then suppose there is $i \in J$ such that $\min\{2 (x^i)^\top H r^i + c^\top r^i : x^i \in P + F^I_i\}$ is negative. From  Claim~\ref{claim one}, let $\tilde{x}^i \in P + F^I_i$ have \smallc such that $\tilde{v}:= 2 (\tilde{x}^i)^\top H r^i + c^\top r^i < 0$; notice that $\tilde{v}$ has \smallc. Given \eqref{eq:linearRj}, we set $\mu = \lceil \frac{(\tilde{x}^i)^\top H \tilde{x}^i + c^\top \tilde{x}^i + d}{|\tilde{v}|} \rceil$ to get that $\tilde{x}^i + \mu r^i$ is feasible for the problem \eqref{eq:lastPart} and has \smallc; this concludes the proof in this case.
	
	Finally, consider the case where for all $i \in J$ we have $\min\{2 (x^i)^\top H r^i + c^\top r^i : x^i \in P + F^I_i\}$ non-negative. In this case, problem \eqref{eq:lastPart} is feasible if and only if 
	\begin{align} \label{eq:lastPart2}
		\begin{split}
		&\qu(x) \le 0\\
		& x \in P + F^I_J
		\end{split}
	\end{align}
	is feasible.
	
	First we bound the norm of solutions to the above problem. 
	
	\begin{claim} \emph{2} \label{claim two}
		There is a rational number $v^*$ of \smallc such that for all $x$ satisfying \eqref{eq:lastPart2} we have $\|x\| \le v^*$.
	\end{claim}
	
	\begin{cpf}
	Let $\h = \{x : f^\top x = 1\}$ be the normalizing hyperplane given by Lemma \ref{lem:hyperplane} for the cone $F$. Now consider any point of the form $\bar{p} + \bar{r} \in P + F^I_J$ (with $\bar{p} \in P$ and $\bar{r} \in F^I_J)$ such that $\qu(\bar{p} + \bar{r}) \le 0$; also consider the vector in direction $\bar{r}$ that belongs to $\h$, namely let $\bar{r} = \lambda \tilde{r}$ for $\tilde{r} \in F_J \cap \h$ and $\lambda > 0$. We upper bound the norm of $\bar{p} + \bar{r}$, starting by bounding $\lambda$.
	
	Since $F$ satisfies the conclusion of Lemma \ref{decomposition}, and given the definition of $J$, we have that $r^\top H r > 0$ for all $r \in F_J$. Let $v^*_1 = \min \{ r^\top H r : r \in F_J \cap \h\}$ (notice that $F_J \cap \h$ is compact). Since $F_J, \h$ and $H$ have \smallc, it follows from Theorem \ref{QPisNP} that
%	$r^*$ and $v^*$ also have \smallc.
$v^*_1$ also has \smallc. 
%\mmnote{Check interface with theorem} 
Evaluating $\qu$ over $\bar{p} + \bar{r}$ we have
	\begin{align} \label{eq: last}
		\qu(\bar{p} + \bar{r}) = \lambda^2 \tilde{r}^\top H \tilde{r} + \lambda \left( 2 \tilde{r}^\top H \bar{p} + c^\top \tilde{r} \right) + \left( (\bar{p})^\top H \bar{p} + c^\top \bar{p} + d \right).
	\end{align}
	Let $v^*_2 := \min \{2 r^\top H p + c^\top r : p \in P, r \in F_J \cap \h\}$ and $v^*_3 := \min\{ p^\top H p + c^\top p + d : p \in P\}$, so that $\qu(\bar{p} + \bar{r}) \ge  \lambda^2 v^*_1 + \lambda v^*_2 + v^*_3$.
	Since $v^*_1 > 0$, the polynomial $\lambda^2 v^*_1 + \lambda v^*_2 + v^*_3$ is strictly convex (as a function of $\lambda$), and since $\qu(\bar{p} + \bar{r}) \le 0$, we have that $\lambda$ cannot be larger than the largest of its roots; explicitly, $\lambda \le \bigg\lceil \frac{-v^*_2 \pm \sqrt{(v_2^*)^2 - 4 v^*_1 v^*_3}}{2 v^*_1} \bigg\rceil$. Moreover, this bound is independent of our choice of point $\bar{p} + \bar{r}$ and is a \smallc value. 
	
	We can finally bound the norm of $\bar{p} + \bar{r}$. By triangle inequality, $\|\bar{p} + \bar{r}\| \le \|\bar{p}\| + \lambda \|\tilde{r}\|$.
	%Letting $v^*_4$ be the ceiling of $\max_i \max \{(e^i)^\top p : p \in P\}$, we can bound the first term $\|\bar{p}\| \le \sqrt{n} \cdot v^*_4$.
	Let $v^*_4$ be the ceiling of $||P||_\infty$. $v^*_4$ has \smallc, because it can be obtained as the ceiling of $\max_i \{ |\max \{(e^i)^\top p : p \in P\}|, |\min \{(e^i)^\top p : p \in P\}|\}$.
	Therefore we can bound $\|\bar{p}\| \le \sqrt{n} \cdot v^*_4$.
	%Letting $v^*_4$ be the ceiling of $||P||_\infty$, we can bound the first term $\|\bar{p}\| \le \sqrt{n} \cdot v^*_4$.
	Also, since $\tilde{r} \in \h$, and by the definition of $\h$ (see Lemma \ref{lem:hyperplane}), we have $f^\top \tilde{r} = 1$ and $f^\top \frac{\tilde{r}}{\|\tilde{r}\|} \ge \frac{1}{\max_j \|r^j\|}$, which imply $\|\tilde{r}\| \le \max_j \|r^j\|$. Together, these bounds give an upper bound for $\|\bar{p} + \bar{r}\|$ by a \smallc value which is independent of $\bar{p} + \bar{r}$; this concludes the proof.
	\end{cpf}
	
	Now we show that if \eqref{eq:lastPart2} has a feasible solution, then it has one of \smallc. Let $\bar{x}$ be a solution for  \eqref{eq:lastPart2} and recall that $\bar{x} \in \Z^p \times \R^q$. By the bound of Claim~\ref{claim two}, and using integrality, we have that the first $p$ components of $\bar{x}$ have \smallc. Then we fix these value and consider the optimization over the other components $\min\{ \qu(x) : x \in P + F^I_J, ~x_i = \bar{x}_i ~\forall i \le p\}$. Claim 2
%~\ref{claim two} 
implies that this optimization problem has a global optimal solution and therefore from Theorem \ref{QPisNP}, we know that this optimization problem has an optimal solution $\tilde{x}$ that has \smallc, and by definition $\tilde{x} \in P + F^I_J$ and $\qu(\tilde{x}) \le \qu(\bar{x}) \le 0$, and hence $\tilde{x}$ is the desired feasible solution for \eqref{eq:lastPart2}. This concludes the proof.   \hfill $\square$
\end{proof}

%##################################################################
%##################################################################
%##################################################################
%##################################################################

\ifmp
	\bibliographystyle{spmpsci}  % mathematics and physical sciences
\else
	\bibliographystyle{amsplain}     
\fi

\bibliography{ip}

\end{document}